\documentclass[journal]{IEEEtran}

\ifCLASSINFOpdf
\usepackage{subcaption}
  \usepackage[pdftex]{graphicx}
  \usepackage{multicol}
  \graphicspath{{./figs/}}
\else
\fi
\usepackage[cmex10]{amsmath}
\usepackage{amsfonts,amssymb,amsthm}
\usepackage{algorithm}
\usepackage[noend]{algpseudocode}
\newtheorem{lemma}{Lemma}
\newtheorem{definition}{Definition}

\usepackage{tikz}
\usepackage{textcomp}
\usepackage{hyperref}
\begin{document}
\title{Mobility-Aware Resource Allocation  in VLC Networks Using T-Step Look-Ahead Policy}

\author{Mohammad~Amir~Dastgheib,
	Hamzeh~Beyranvand,~\IEEEmembership{Member,~IEEE,}
	Jawad~A.~Salehi,~\IEEEmembership{Fellow,~IEEE}, and~Martin~Maier,~\IEEEmembership{Senior Member,~IEEE}
\thanks{M.~A.~Dastgheib and J.~A.~Salehi are with the Department of Electrical Engineering, Sharif University of Technology, Tehran, Iran (e-mail: sma.dastgheib@ee.sharif.edu, jasalehi@sharif.edu).}
\thanks{H.~Beyranvand is with the Department of Electrical Engineering, Amirkabir University of Technology, Tehran, Iran (e-mail: beyranvand@aut.ac.ir).}%
\thanks{M. Maier is with the Optical Zeitgeist Laboratory, Institut National de la
Recherche Scientifique (INRS), Montreal, QC H5A 1K6, Canada (e-mail:
maier@ieee.org). }%
\thanks{ \textcopyright 2018 IEEE.  Personal use is permitted, but republication/redistribution requires IEEE permission. See \href{http://www.ieee.org/publications standards/publications/rights/index.html}{www.ieee.org} for more information.
	
	DOI: \href{https://doi.org/10.1109/JLT.2018.2872869}{10.1109/JLT.2018.2872869}}
}


\maketitle

\begin{abstract}
Visible light communication (VLC) uses  huge license-free spectral bandwidth of visible light for high-speed wireless communication. Since each VLC access point covers a small area, handovers of mobile users are inevitable. In order to deal with these handovers, developing fast and effective resource allocation algorithms is a challenge. This paper introduces the problem of mobility-aware optimization of resources in VLC networks using the \textit{T-step look-ahead} policy and employing the notion of \textit{handover eff\/iciency}. It is shown that the handover eff\/iciency can correlate the overall performance of the network with future actions based on the mobility of users. Due to the stationary nature of indoor optical wireless channels, future channel state information (CSI) can be predicted by anticipating the future locations of users. A resource allocation algorithm is thus proposed, which uses CSI to dynamically allocate network resources to users. To solve this mathematically intractable problem, a novel relaxation method is proposed, which proves to be a useful tool to develop efficient algorithms for network optimization problems with a proportional fairness utility function. The resulting algorithm is extremely faster than the previous method and awareness of mobility enhances the overall performance of the network in terms of rate and fairness utility function.
\end{abstract}


\IEEEpeerreviewmaketitle

\section{Introduction}
\IEEEPARstart{A}{ccording} to network analytics, about 80\% of wireless data usage is by indoor users \cite{cisco2012indoor}. Thus, increasing the transmission rate of indoor wireless access networks is of paramount importance. Visible light communication (VLC) is a promising solution for next-generation indoor access networks, which leverages illumination and data transmission simultaneously. Utilizing huge license-free spectral bandwidth of visible light, VLC offers massive data rates \cite{tsonev2015rate}, while satisfying energy-eff\/iciency and security requirements of next-generation networks \cite{buzzi2016survey}. Due to its immunity to radio frequency (RF) interference, VLC may be used to effectively offload heavy traff\/ic loads from RF wireless networks.

In recent studies on VLC networks, there has been an interest in evaluating resource allocation algorithms for mobile users \cite{wang2015dynamic,wang2017optimization}. A rather different approach is to assume \textit{a prior} knowledge of mobility in developing new algorithms \cite{chiu2000predictiveblocking,li2016mobility,wang2017mobility,elazzouni2018qos}. Despite the fact that mobility increases the capacity of an ad-hoc wireless network \cite{tse2001mobility}, mobile users may have some defects on the performance of a cellular wireless network. More specif\/ically, mobility may induce handovers, which cause a degradation in users' data rate and result in users' connections being blocked. To deal with mobile users, having some information about users' mobility may help to develop algorithms with an improved performance for mobile networks. For example, the blocking probability can be reduced in a cellular network by sending reservation requests to neighboring access points (APs) based on a user's mobility pattern \cite{chiu2000predictiveblocking}. This approach suggests that a mobility-aware perspective may be useful in developing VLC network algorithms.

\subsection{Related Work}
The mobility pattern of users may be considered a part of \textit{system events} that will happen in the future. Hence, the problem of mobility-aware optimization coincides with the general problem of \textit{T-step look-ahead policy} introduced by Neely in the Lyapunov optimization framework \cite{neely2010universal}. However, since future system events are unknown in general, \cite{neely2010universal} uses the optimum answer to this problem as a benchmark to evaluate a \textit{non-anticipating} algorithm. Following this perspective  \cite{elazzouni2018qos} proposed a mobility-aware resource allocation in the Discounted Rate Utility
Maximization (DRUM) framework \cite{eryilmaz2017discounted}. This algorithm requires the knowledge of time-varying capacities of all radio interfaces and is applicable if the wireless channel is predictable for some specif\/ic duration. The mobility-aware design is also used to obtain cache placement algorithms for Device to Device (D2D) communications, maximizing the data offloading ratio by considering users’ contact
and inter-contact durations \cite{wang2017mobility,wang2018exploiting}. 

In VLC networks, future system events with regard to user mobility can be predicted. The f\/irst study that uses mobility awareness in the resource allocation of VLC networks is \cite{li2016mobility}, which uses a graph-based approach. In \cite{li2016mobility}, the problem of determining the access point that should serve a user is modeled by a college admission problem. The authors showed that using a proper preference function, the users' traff\/ic load can be balanced among APs. In \cite{li2016mobility}, only the knowledge of the movement direction of mobile users has been assumed. Furthermore, the authors of \cite{li2016mobility} neither use a proportional fairness utility function nor solve any optimization problem. Following the general look-ahead policy of \cite{neely2010universal}, an anticipatory algorithm was proposed in \cite{zhang2018anticipatory}, which utilizes the predictability of VLC channel to ensure queue stability by using a Largest Weighted Delay First (LWDF) scheduling technique \cite{stolyar2001largest}. Instead of LWDF or DRUM, in this paper, we formulate the mobility-aware problem using a rather different approach based on handover eff\/iciency. We show that the resulting algorithm enhances the overall network utility even when there are only two VLC APs available.

Developing eff\/icient resource allocation schemes in VLC networks is a challenge that has been widely studied recently. Most of the studies ignore the possibility of using a prior knowledge of users' mobility in developing such algorithms. In this scope, different constraints have been considered in the optimization of load balancing in such networks. In \cite{jin2015resource}, the metric of \textit{effective capacity} \cite{wu2003effective} has been considered  to satisfy given transmission delay constraints. The authors showed that by optimizing a function of the effective capacity, while taking fairness into account, the given stochastic delay constraints are satisf\/ied. The same fairness function has been utilized in \cite{wang2017optimization} to optimize users' data rates. In \cite{wang2015dynamic}, a logarithmic fairness function, which is a special case of the general proportional fairness function used in \cite{wang2017optimization} and \cite{jin2015resource}, has been used to optimize users' data rates. In \cite{wang2015dynamic}, Wang and Haas modeled the handover penalty in terms of handover eff\/iciency, which shows the reduction of a user's data rate due to handover. The resource allocation or load balancing problem has been then formulated using this notion in \cite{wang2015dynamic} and \cite{wang2017optimization}.

There exist different approaches to solve the arising optimization problem of the resource allocation in VLC networks. The proposed optimization algorithm should be fast enough to be able to run online. Hence, major efforts have been devoted to the development of fast algorithms after relaxation of the binary allocation variables. When the fairness function is logarithmic, discretizing the scheduling time period is useful \cite{wang2015dynamic,li2015cooperative}. However, in general, the most widely used approach is the Lagrangian multiplier method \cite{wang2017optimization,jin2015resource,kashef2016energy}. Specif\/ically, the general proportional fairness function is usually applied via a decomposition method \cite{chiang2007decompose}, as done in \cite{wang2017optimization,jin2015resource}, before applying the Lagrangian multiplier method, which is time consuming \cite{wang2017optimization}. The reason is the need for keeping the allocation variables binary at the end of each iteration, which results in many comparisons. In this paper, we propose a novel relaxation method that resolves this issue and increases the speed of the algorithm signif\/icantly.

\subsection{Contributions and Outline}
This paper makes the following three main contributions. First, the problem of mobility-aware optimization of resource allocation is introduced using the notion of handover eff\/iciency. As we will show, this notion correlates the overall performance of the network with future actions based on the mobility of users. We show how to use predicted CSI to form an optimization problem, which enhances the overall network utility. Note that the non-mobility aware scheme is a special case of the proposed general problem. The exact solution to this intractable problem is provided using an exhaustive search. Second, we introduce a novel relaxation method, which is benef\/icial for the general optimization with proportional fairness as the objective function. This method relaxes the mathematically intractable problem into a convex one. Due to the proposed convex relaxation method, the resulting algorithm has a closed form solution in each iteration, which increases the speed of the algorithm signif\/icantly. The last contribution of this paper is the analysis of an upper bound on the service time interval.

For real-time applications, achieving a high data rate for all users is crucial during the whole duration of service times. Hence, in this paper the notion of fairness is extended from \textit{fairness amongst users} to \textit{fairness amongst different users at different service times}. Another way to explain this is in terms of \textit{average utility function}. The overall performance of the network in the long term is more important than its instantaneous performance. Hence, it is benef\/icial to design an algorithm that optimizes the utility function of the network in the long term of the system rather than during short service time intervals. To optimize the long-term performance of VLC networks, our proposed algorithm considers the utility function over future service time intervals. 

In the absence of fading, the optical channel of VLC networks is almost stationary. Hence, the knowledge of the future velocity and position of users can be used to predict the signal to interference and noise ratio (SINR) or the CSI between users and each AP. We utilize this prediction in our optimization problem. To solve the problem, we f\/irst relax the handover eff\/iciency of \cite{wang2017optimization} and replace it with a linear function to render the relaxation of the general problem feasible. Second, the binary allocation variable is relaxed in a novel way suchthat the objective function of the problem is converted into a convex one. Two lemmas provide the necessary and suff\/icient conditions for this relaxation method, leading to convex/concave objective functions. After this relaxation, the standard Lagrangian multiplier method is used to iteratively solve the problem. For a fair comparison, our proposed algorithm is compared with the decomposition algorithm of \cite{wang2017optimization}, which uses the same general fairness utility function. 

The remainder of the paper is organized as follows. Section  \ref{sec:system} describes the system model, including the notion of handover eff\/iciency, and gives an intuitive example to illustrate the benef\/its of using mobility awareness in optimizing network utility. In Section III, we explain our proposed mobility-aware resource allocation in greater detail. Section \ref{sec:mobility} describes the mobility model used in our verifying simulations. An upper bound on service time intervals is also derived in this section. Section \ref{sec:numerical} presents numerical results and evaluates the performance of the proposed algorithm. Finally, Section \ref{sec:conclusion} concludes the paper.
\section{System Model}\label{sec:system}

\subsection{System Overview}
In this paper, an indoor multi-user access environment is considered. This environment is formed by $|\mathcal{A}|$ VLC APs and $|\mathcal{U}|$ users, where $\mathcal{A}$ and $\mathcal{U}$ denote the set of APs and users, respectively. VLC APs utilize light emitting diodes (LED) for communication. Because of the nature of LED light, communication is restricted to a conf\/ined area. Each user, has a photodetector (PD) that is assumed to be oriented perpendicular to the floor.  Time is partitioned into intervals of duration $\tau_p$, referred to as service time. A user connects to one of the VLC  APs in each service time interval for downlink communication. The uplink may be realized by using WiFi or VLC and is not discussed here. Users may move between APs, whereby their movement may initiate a handover from one VLC AP to another, which is inevitable and degrades the performance of the system. All APs are connected to a central control unit (CU), which controls AP assignment and chooses the right time to initiate a handover. A resource allocation algorithm runs prior to each service time interval and determines the interconnections between users and APs in the forthcoming service time interval. 

Further, we assume that a positioning system is available, which is used to measure the position of each mobile user. The measured positions are later used by the CU to predict the position of each user in future service time intervals. A VLC positioning system utilizing the available VLC APs \cite{zabih2017experimental} with an accuracy ``\textit{up to the level of centimeters}" \cite{zhang2014theoretical} is a good candidate for this purpose. The CU uses the predicted positions to determine future CSI.

As an initial work, this paper do not consider the vertical tilting of receiver's PD. However the orientation of users' PD may affect both positioning and handover. Particularly, the achievable throughput from an AP to a specif\/ic user is a function of the orientation of that user's PD \cite{pakravan2001indoor}. As a result, cell's coverage will also be dependent on the orientation of mobile users' PD \cite{ghassemlooy2017visible}. Most mobile devices are equipped with sensors that enable them to calculate orientation of the device. These calculated orientations can help to f\/ind the real achievable rate and correct position of each user, which can be used later by the resource allocation algorithm \cite{soltani2016access}. Rotation of users' receivers may induce additional handovers \cite{soltani2017handover}. While dealing with the effects of PD orientation and rotation is an open research area, the scope of this paper is to investigate the effect of users' mobility. Hence we focus on the mobility and leave the more thorough study of orientation and rotation for future researches.

For each quantity (rate, allocation, etc.), subscripts $\alpha$  and  $\mu$ denote the index of a specif\/ic AP and user, respectively. Superscript $t$ shows the value of the quantity at $t$ service times later than the current service time. Note that $t$ is the index of future service time and hence is unit-less. As an example, the total predicted rate that AP $\alpha$ can provide for user $\mu$ at $t$ service times later is denoted by $R^t_{\mu,\alpha}$. This value is a function of the user's predicted position.

\subsection{VLC Channel Model}\label{sec:channel}
Visible light communication uses intensity modulation at the transmitter. A photo-detector at the receiver detects the modulated signal. Due to the reflection of light from walls, the received optical power is the sum of Line-of-Sight (LoS) and non-LoS components. According to \cite{komine2004fundamental}, the received optical power from the reflections is negligible. Hence, similar to \cite{shao2016joint,li2016mobility}, this paper considers only the LoS path in the channel model. The VLC channel path loss is then obtained as follows \cite{kahn1997wireless,ghassemlooy2017visible}:
\begin{equation}
H_{\mu,\alpha}^t=\left\lbrace\begin{array}{ll}
\frac{(m+1)A}{2\pi D_d^2}\cos^m(\phi)T_s(\psi)g(\psi)\cos(\psi) & 0\leq\psi\leq\Psi_f,\\
0 & \psi>\Psi_f,
\end{array}\right.
\end{equation}
where $m$ is the Lambertian order, which is a function of the half-intensity radiation semi-angle $\Phi_{1/2}$ according to $m=-{1}/{\log_2(\cos(\Phi_{1/2}))}$. Further, $A$ is the physical area of the receiver photo-diode (PD), $D_d$ is the distance between VLC AP and a user's PD, $\Psi_f$ is the receiver Field of View (FoV), $T_s(\psi)$ is the gain of the optical filter and $g(\psi)$ is the concentrator gain, which for a reflection index of $n$ is given by
\begin{equation}
g(\psi)=\left\lbrace\begin{matrix}
\displaystyle\frac{n^2}{\sin^2(\Psi_f)} & 0\leq\psi\leq\Psi_f,\\
0 & \Psi_f\geq\psi.
\end{matrix}\right.
\end{equation} 

As shown in Fig. \ref{fig:angles}, we have $D_d=\sqrt{r^2+h^2}$, $\cos(\phi)=\cos(\psi)=\frac{h}{D_d}$, and $r=\|\vec{y}_\mu^{\,t}-\vec{y}_\alpha\|_2 $, whereby $\vec{y}_\mu^{\,t}$ is the position of user $\mu$ at service time $t$, $\vec{y}_\alpha$ is the position of AP $\alpha$, and $\|\cdot\|_2$ denotes norm-2. Hence, the channel path loss inside FoV of the receiver is as follows:
\begin{equation}
H_{\mu,\alpha}^t=
\frac{(m+1)An^2 h^{m+1}}{2\pi \sin^2(\Psi_f)}\left(\|\vec{y}_\mu^{\,t}-\vec{y}_\alpha\|_2^2+h^2\right)^{-\frac{m+3}{2}}.
\end{equation}
\begin{figure}[t]
	\centering
	\includegraphics{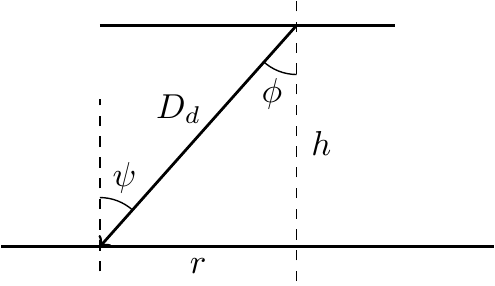}
	\caption{Angle of incidence ($\psi$) and irradiance ($\phi$).}
	\label{fig:angles}
\end{figure}

The SINR of the VLC channel is computed as follows \cite{wang2017optimization}: 
\begin{equation}
\text{SINR}_{\mu,\alpha}^t=\frac{\kappa^2(P H_{\mu,\alpha}^t)^2}{\displaystyle \iota^2NB+\kappa^2\sum_{\alpha'\in S_i}(PH_{\mu,\alpha'}^t)^2}\,,
\end{equation}
where $P$ is the transmitted optical power, $\kappa$ denotes the optical-to-electrical conversion eff\/iciency, $N$ is the power spectral density of receiver noise, $B$ is the bandwidth, and $\iota$ is the ratio of average transmitted optical power and square root of the electric power of signals without DC bias. In case of $\iota=3$, clipping noise is negligible \cite{wang2017optimization}. Further, $S_i$ is the set of APs that may cause interference at the receiver of user $\mu$. The available data rate between user $\mu$ and AP $\alpha$ is then given by the Shannon capacity:
\begin{equation}\label{eq:capacity}
R_{\mu,\alpha}^t=B\log_2(1+\text{SINR}_{\mu,\alpha}^t).
\end{equation}

\subsection{Modeling Handover Overhead}\label{sec:hand}
The notion of handover eff\/iciency was introduced in \cite{wang2015dynamic} and was later used in \cite{wang2017optimization} to represent the negative effect of handover on users' data rate. Assuming $\eta_0$ to be the average handover eff\/iciency, the available data rate between user $\mu$ and AP $\alpha$ can be computed as 
\begin{equation}\label{eq:rate}
r_{\mu,\alpha}^t=\eta_{\mu,\alpha}^t( x_{\mu,\alpha}^{t-1})R_{\mu,\alpha}^t=\left\lbrace\begin{matrix}
\eta_0 R_{\mu,\alpha}^t &  x_{\mu,\alpha}^{t-1}=0,\\
R_{\mu,\alpha}^t & x_{\mu,\alpha}^{t-1}=1,
\end{matrix}\right.
\end{equation}
where $x_{\mu,\alpha}^{t-1}$ is a binary variable that indicates the AP assignment. $x_{\mu,\alpha_0}^{t}$ equals one, if user $\mu$ is connected to AP $\alpha_0$ at service time $t$ and is zero otherwise. Since each user can be connected to one AP at each service time, for each user only one of the binary variables $x_{\mu,\alpha}^{t}$ would be one and others will be zero.
\begin{equation}
x_{\mu,\alpha_0}^{t}=1\Rightarrow x_{\mu,\alpha}^{t}=0,\, \forall \alpha\in\mathcal{A},\alpha\neq\alpha_0
\end{equation}
Note that $x_{\mu,\alpha}^{0}$ is not a variable since it is def\/ined as the AP assignment at the previous service time. The above handover eff\/iciency can be expressed as a linear function of $x_{\mu,\alpha}^{t-1}$ as follows: 
\begin{equation}\label{eq:efficiency}
\eta_{\mu,\alpha}^t(x_{\mu,\alpha}^{t-1})=(1-\eta_0)x_{\mu,\alpha}^{t-1}+\eta_0,
\end{equation}
which has the same value as (\ref{eq:rate}) for $x_{\mu,\alpha}^{t-1}=0,1$. This linear function will be used as a relaxation tool in Section \ref{sec:allocation}.

Ping-pong effect, i.e. alternating handover initiation between adjacent APs, is an important issue in cellular communications. In VLC networks, due to the stationary nature of indoor optical wireless channel, the ping-pong effect is mostly the result of occasional changes in the orientation of PD \cite{soltani2017handover} or blocking of LoS channel by obstacles \cite{wang2015efficient}.  In such cases there exists a trade-off between \textit{waiting} for a certain time for possible link recovery or \textit{switching} to another AP to prevent communication breakdown. In this study we rely on optimization to deal with handovers. The notion of handover efficiency  reduce the preference of resource allocation algorithm to alternate users between adjacent APs due to decrease in rate. Hence, it may reduce the ping-pong effect.
\subsection{Intuitive Rationale Behind Using Mobility Awareness}
As an illustrative example, consider a simple 2-user scenario, in which 2 users are assumed to be associated with AP1 at their previous service time. The scenario along with the users' mobility path are depicted in Fig. \ref{fig:intuition}. Each user is moving toward its neighboring AP. There exist 2 zones in each AP. For simplicity, the available data rate of the APs is approximated as follows:
\begin{figure}[bt]
	\centering
	\includegraphics[width=0.5\textwidth]{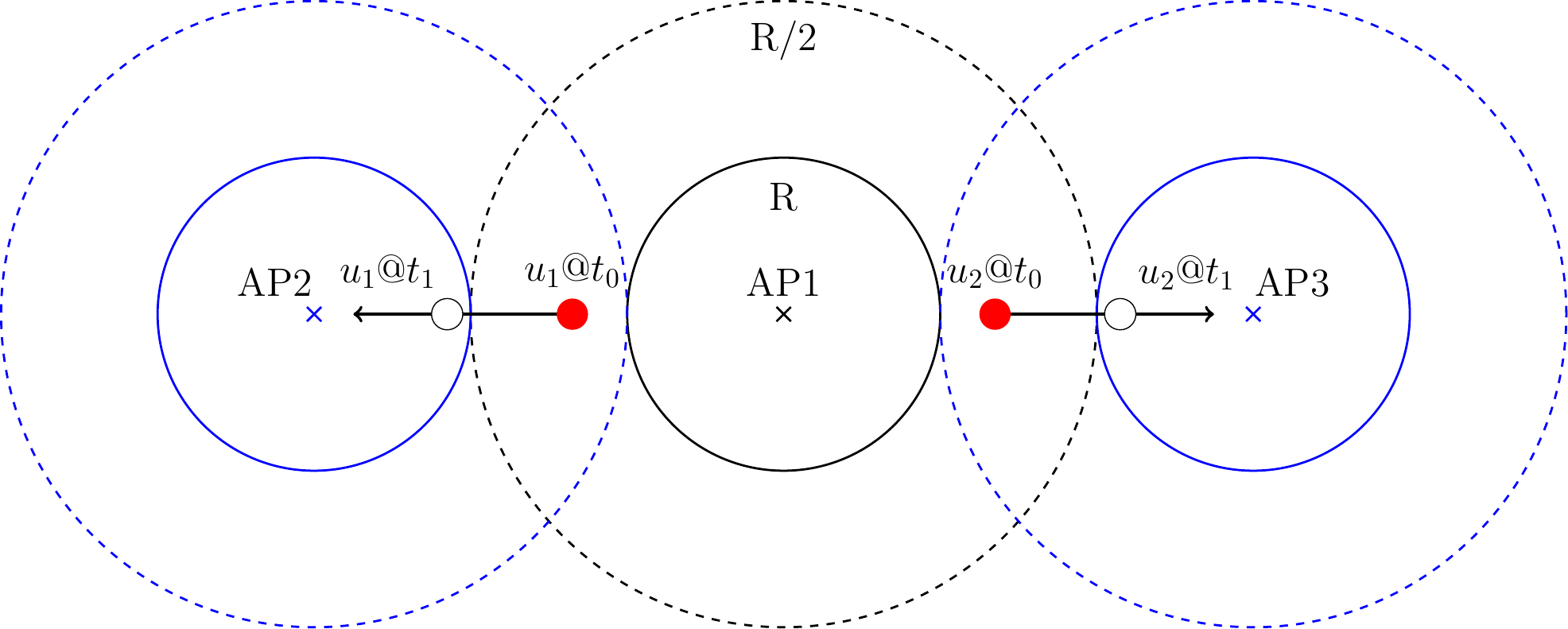}
	\caption{A simple 2-user example illustrating the intuition behind the presented mobility-aware approach.}
	\label{fig:intuition}
\end{figure}
\begin{itemize}
	\item In the zone, which is nearer to the AP (inside the smaller circle), users may be served with a total rate equal to $R$.
	\item In the second zone (between the circles), users may be served by a lower rate, say $R/2$, due to increased distance.
\end{itemize}
  
\begin{table}[bt]
	\caption{Possible Choices for Handover Initiation}
	\centering
	\begin{tabular}{|c|c|c|c|c|}
		\hline
		\multicolumn{2}{|c|}{1st service-time} & \multicolumn{2}{c|}{2nd service-time}& \\
		\hline
		\# & Rate & \# & Rate & Average\\
		\hline
		&&&&\\[-1em]
		0 & $\frac{R}{4}+\frac{R}{4}$ & 2 & $\eta_0R+\eta_0R$ & $\frac{1}{2}\left(\frac{1}{2}+2\eta_0\right)R$\\
		&&&&\\[-1em]
		1 & $\eta_0\frac{R}{2}+\frac{R}{2}$& 1 & $R+\eta_0 R$ & $\frac{1}{2}\left(\frac{3}{2}+\frac{3}{2}\eta_0\right)R$\\
		&&&&\\[-1em]
		2 & $\eta_0\frac{R}{2}+\eta_0\frac{R}{2}$& 0 & $R+ R$ & $\frac{1}{2}(2+\eta_0)R$\\
		\hline
		\multicolumn{5}{l}{(\# denotes the number of handovers at the corresponding service time.)}
	\end{tabular}\label{tab:intuition}
\end{table}
In the considered example, the objective is to maximize the users' aggregate data rate. The f\/irst (current) service time starts at $t_0$ and the second one starts at $t_1$.  Since both users are outside the coverage area of AP1 at the second service time, each user should be disconnected from AP1 and connect to another AP at the f\/irst or second service time. Thus, the algorithm has to make choice such that for each user a handover should be initiated at the f\/irst or second service time. There are 3 possible choices summarized in Table \ref{tab:intuition}. Since each user should experience a handover at one of the aforementioned service times, each choice at the f\/irst service time enforces the allocation of the second one. Toward this end, two different schemes are considered:
\subsubsection{Mobility-unaware scheme}
In this scheme, the algorithm is unaware of users' mobility. Hence, it opts to maximize the data rate at the current service time. In this approach, the second scenario (row 2) of Table \ref{tab:intuition} is selected due to the higher data rate at f\/irst service time, in which the average data rate is given by
$$R_{\text{avg}}=\frac{1}{2}\left(\frac{3}{2}+\frac{3}{2}\eta_0\right)R.$$

\subsubsection{Mobility-aware scheme}
In the second scheme, the algorithm knows in which region users $u_1$ and $u_2$ are and to which APs they can be connected to be served at which rates at each service time. Thus, it optimizes the aggregate data rate over several service times. This means that the algorithm looks at the \textit{Sum} column of Table \ref{tab:intuition} and chooses the handover strategy that maximizes the aggregate data rate, which is the objective in this simplif\/ied example. Here, initiating the handover at time $t_0$ for both users is the optimum strategy, if the prediction is done only for one future service time, resulting in overall average throughput given by
$$R_{\text{avg}}=\frac{1}{2}(2+\eta_0)R,$$
which is higher than that of the mobility-unaware scheme. 

The intuitive rationale behind this increase is that the f\/irst algorithm is unaware of future movements. Since users are moving toward an AP, which provides them with higher data rates in future, it is better to associate them with these APs sooner. 

The above 2 zone approximation is only for the purpose of illustrating the rationale behind using mobility awareness and the rest of the paper is based on the exact model of VLC channel given in subsection \ref{sec:channel}. Indeed the approximation can become more accurate by considering more than 2 zones in each AP.

Analogous to the explained simplified example, we may obtain graphs showing the real evolution of users' rates over time. These graphs are plotted in Fig. \ref{fig:intuition2}. Each plot is equivalent to one column of Table \ref{tab:intuition} and different curves in each of these plots equals different rows of Table \ref{tab:intuition}. To obtain the plots we rely on the definition of handover efficiency. Due to the signaling overhead in case of handover initiation, the user's service will be interrupted for a certain amount of time \cite{van2012handoff}. As discussed earlier in subsection \ref{sec:hand}, this process can be modeled using handover efficiency in terms of a reduction in the user's average rate. The real curve of user's rate vs. time can be plotted by considering the service interruption time. For example if we consider $\eta_0=0.75$, then in case of handover initiation, in the first 25\% of a service time user's service will be interrupted. Using this method we may obtain the curves in Fig. \ref{fig:intuition2}.  The average rate corresponding to each handover scenario is calculated and reported in Table \ref{tab:intuition2}. Again we may observe that mobility-unaware scheme will select the middle row while the mobility-aware scheme will select the last row resulting in higher overall average rate for users.      

\begin{figure}
	\begin{subfigure}{.45\linewidth}
		\centering
		\includegraphics[width=\linewidth]{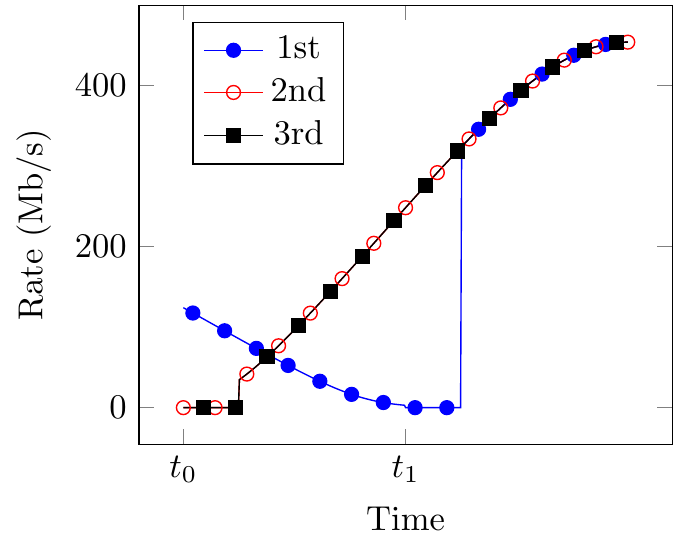}
		\caption{Rate of user 1}
		\label{fig:sub1}
	\end{subfigure}%
	\begin{subfigure}{.45\linewidth}
		\centering
		\includegraphics[width=\linewidth]{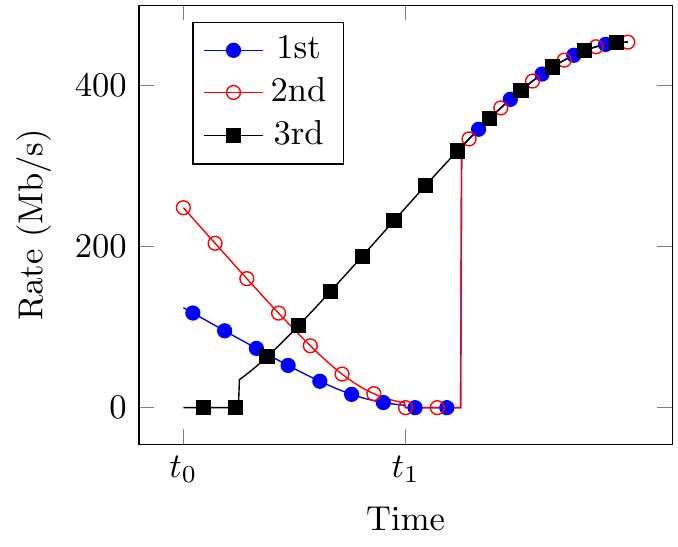}
		\caption{Rate of user 2}
		\label{fig:sub2}
	\end{subfigure}\\[1ex]
	\begin{subfigure}{0.9\linewidth}
		\centering
		\includegraphics[width=0.8\linewidth]{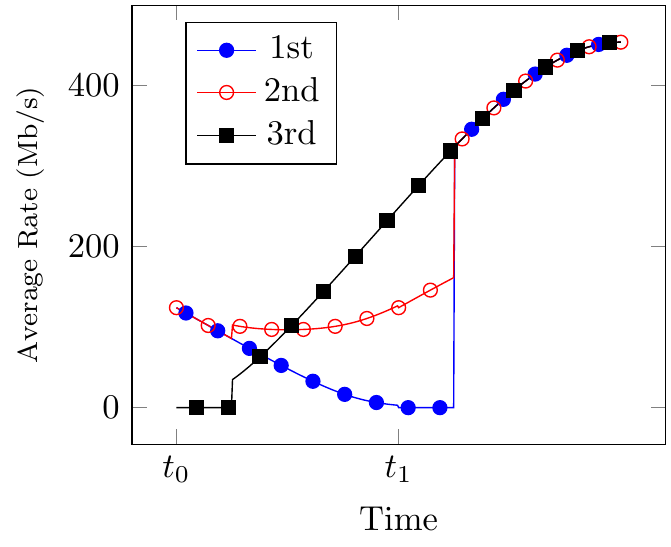}
		\caption{Average of rates}
		\label{fig:sub3}
	\end{subfigure}
	\caption{Effect of different handover scenarios on users' rate with the exact model for VLC channel using the notion of handover efficiency.}
	\label{fig:intuition2}
\end{figure}

\begin{table}[bt]
	\caption{Possible Choices for Handover Initiation using Exact Channel Model}
	\centering
	\begin{tabular}{|c|c|c|c|c|}
		\hline
		\multicolumn{2}{|c|}{1st service-time} & \multicolumn{2}{c|}{2nd service-time}& \\
		\hline
		\# & Rate (Mb/s) & \# & Rate (Mb/s) & Average (Mb/s)\\
		\hline
		&&&&\\[-1em]
		0 & 106.5 & 2 & 609.9 & 358.8\\
		&&&&\\[-1em]
		1 & 207.6 & 1 & 681.7 & 445.3\\
		&&&&\\[-1em]
		2 & 202.3 & 0 & 753.6 & 478.6\\
		\hline
	\end{tabular}\label{tab:intuition2}
\end{table}

\section{Optimal Resource Allocation}
\label{sec:allocation}

\subsection{Problem Formulation}
Choosing a proper utility function is crucial in the optimization of networks. A widely used utility function is $\beta$-proportional fairness and for $\beta\neq 1$ it is def\/ined as \cite{mo2000fair}
\begin{equation}\label{eq:beta_fair}
\psi_\beta(r)=\frac{r^{1-\beta}}{1-\beta}=\frac{-1}{(\beta-1)r^{\beta-1}},
\end{equation}
where $r$ denotes the data rate. This function has a positive value for $\beta<1$ and a negative value for $\beta>1$.

Using the utility function of (\ref{eq:beta_fair}) and the relaxed notion of handover eff\/iciency in (\ref{eq:efficiency}), the optimization problem can be written as follows:
\begin{subequations}
\begin{alignat}{2}
&\max_{x_{\mu,\alpha}^t,p_{\mu,\alpha}^t} \quad \sum_{\mu}\sum_{\alpha}\sum_{t}x_{\mu,\alpha}^t\psi(p_{\mu,\alpha}^t r_{\mu,\alpha}^t)\tag{10}\label{eq:mobilityaware}\\
\text{s.t.}&\; \,r_{\mu,\alpha}^t=\left((1-\eta_0)x_{\mu,\alpha}^{t-1}+\eta_0\right)R_{\mu,\alpha}^t \;\text{{\small  $\forall \mu\in\mathcal{U},\forall \alpha\in\mathcal{A}, t\geq 2;$}}\label{eq:mobilityawarerate}\\
&\; \,\sum_{\mu\in\mathcal{U}}x_{\mu,\alpha}^t p_{\mu,\alpha}^t\leq 1 \qquad \forall \alpha\in\mathcal{A}, \,\,t\in\mathbb{N};\label{eq:mobilityawarexp}\\
&\;\,\sum_{\alpha\in \mathcal{A}}x_{\mu,\alpha}^t=1 \qquad\quad\ \; \forall \mu\in\mathcal{U},\,\,t\in\mathbb{N};\label{eq:mobilityawarex}\\
&\;\, x_{\mu,\alpha}^t\in\lbrace 0,1\rbrace,p_{\mu,\alpha}^t\in\left(0,1\right]\quad\forall \alpha\in\mathcal{A},\,\mu\in\mathcal{U},\,t\in\mathbb{N}.\nonumber
\end{alignat}
\end{subequations}
where $\mathbb{N}=\lbrace 1,2,\ldots\rbrace$ and $p_{\mu,\alpha}^t$ is the proportion of resources that are allocated from AP $\alpha$ to user $\mu$ at $t$-th future service time.

In the above optimization problem, constraint (\ref{eq:mobilityawarexp}) states that each access point can't allocate more than all of its resources and constraint (\ref{eq:mobilityawarex}) states that each user should only be  assigned to one access point. This optimization problem has an infinite number of optimization variables and hence it is impossible to be solved. Since future rates should be predicted before solving the problem, the prediction error would be high for large values of $t$. Hence, $t$ shouldn't go to inf\/inity and should be truncated at some step, say, $t=T$. The parameter $T$ is called \textit{prediction level} in this paper, since it shows how many future service times should be predicted for the algorithm. Similar to $t$, $T$ is unit-less.

The optimization in (\ref{eq:mobilityaware}) is a Mixed Integer Non-Linear Programming (MINLP) problem, which is mathematically intractable.

\subsection{Exhaustive Search}
The exact solution to the MINLP problem may be found using exhaustive search, which is impractical and time consuming. The search should be done on the vector of binary variables. Let $\mathbf{x}$ be the vector of all binary variables at different service times, i.e. $x_{\mu,\alpha}^t\forall \mu\in\mathcal{U},\forall \alpha\in\mathcal{A}, 1\leq t\leq T$, whose elements show the user-AP interconnections. Each binary vector $\mathbf{x}$, called an assignment, corresponds to a certain assignment of users to APs. Then a feasible assignment may be def\/ined as follows.
\begin{definition}[feasible AP assignment]
	 A feasible assignment is def\/ined by a vector $\widehat{\mathbf{x}}$, for which the constraint in (\ref{eq:mobilityawarex}) is satisf\/ied for $1\leq t\leq T$.
\end{definition}
Let $\mathcal{X}$ be the set of all feasible assignments. It is obvious that the search should be done only on feasible assignments. Therefore, at the f\/irst step, the exhaustive search algorithm generates a feasible assignment, namely $\widehat{\mathbf{x}}\in\mathcal{X}$. For this assignment the continuous optimization variables should be determined. This can be done by solving the following optimization problem derived from (\ref{eq:mobilityaware}), keeping in mind that the feasibility of assignment condition is already satisf\/ied.
\begin{alignat}{2}
&\max_{p_{\mu,\alpha}^t} \quad \sum_{\mu}\sum_{\alpha}\sum_{t}\hat{x}_{\mu,\alpha}^t\psi(p_{\mu,\alpha}^t \hat{r}_{\mu,\alpha}^t)\label{eq:mobilityaware2}\\
\text{s.t.}&\; \sum_{\mu\in\mathcal{U}}\hat{x}_{\mu,\alpha}^t p_{\mu,\alpha}^t\leq 1 \qquad \forall \alpha\in\mathcal{A}, \,\,t\in\mathbb{N}^T;\nonumber\\
&\;\, p_{\mu,\alpha}^t\in\left(0,1\right]\quad\forall \alpha\in\mathcal{A},\,\mu\in\mathcal{U},\,t\in\mathbb{N}^T.\nonumber
\end{alignat}
Note that after setting AP assignment, there is no interconnection between future and current optimization variables. Hence the above optimization will be separated into $T$ disjoint optimizations. Since the $\beta$-proportional fairness function is concave, these optimizations are concave optimization problems, which can be solved using the Lagrangian multiplier method. For example, for $\beta>1$, taking into account the fact that objective function is negative, we may write the equivalent convex optimization as:
\begin{alignat}{2}
&\min_{p_{\mu}^t} \sum_{\mu}\sum_{\alpha}\frac{\hat{x}_{\mu,\alpha}^t}{\left(p_{\mu,\alpha}^t \hat{r}_{\mu,\alpha}^t \right)^{\beta-1}}\label{eq:oneAP}\\
&\text{s.t.}\;  \sum_{\mu\in\mathcal{U}}\hat{x}_{\mu,\alpha}^t p_{\mu,\alpha}^t\leq 1 \qquad \forall \alpha\in\mathcal{A};\nonumber\\
&\;\quad\, p_{\mu,\alpha}^t\in\left(0,1\right]\quad\forall \alpha\in\mathcal{A},\,\mu\in\mathcal{U}.
\end{alignat}
Since the objective function is convex and the constraints are aff\/ine, the above problem is a convex optimization problem. The solution to this can be found using Lagrangian multiplier method which is given in \cite{wang2017optimization} as:
\begin{equation}\label{eq:exhpOpt}
p_{\mu,\alpha,\hat{\mathbf{x}}}^{t^*}=\dfrac{\hat{r}_{\mu,\alpha}^{\frac{1}{\beta}-1}}{\displaystyle\sum_{\nu\in\mathcal{U}}\hat{x}_{\nu,\alpha}^t\hat{r}_{\nu,\alpha}^{\frac{1}{\beta}-1} }.
\end{equation}
Then the optimal utility function would be
\begin{equation}
  U_{\hat{\mathbf{x}}}^*=\sum_{\mu}\sum_{\alpha}\sum_{t}\hat{x}_{\mu,\alpha}^t\psi(p_{\mu,\alpha,\hat{\mathbf{x}}}^{t^*} \hat{r}_{\mu,\alpha}^t).
\end{equation}
After calculating the optimal utility for all possible feasible assignments, the assignment that maximizes the utility will be the optimal assignment and the output of the exhaustive search. 
\begin{equation}\label{eq:exhXopt}
\mathbf{x}^*=\arg\max_{\hat{\mathbf{x}}\in\mathcal{X}} U_{\hat{\mathbf{x}}}^*.
\end{equation}

\begin{algorithm}[!htb]
	\caption{Exhaustive Search}\label{alg:exhuast}
	\begin{algorithmic}[1]
		\Statex \textbf{Input}
		\Statex $T$: The total number of future service times influencing current service time.
		\Statex $y_{\mu}^t$: Position of user $\mu$ at service time $0\leq t\leq T$.
		\ForAll {$\mu\in\mathcal{U}$,$\alpha\in\mathcal{A}$,$1\leq t\leq T$}
		\State CU calculates $R^t_{\mu,\alpha}$.
		\EndFor
		\ForAll {$\hat{\mathbf{x}}\in\mathcal{X}$}
		\State Find optimal $\mathbf{p}_{\hat{\mathbf{x}}}^*$ using (\ref{eq:exhpOpt}).
		\State Calculate  $U_{\hat{\mathbf{x}}}^*$
		\EndFor
		\State Find $\mathbf{x}^*$ according to (\ref{eq:exhXopt}).
		\State  $\mathbf{p}^*=\mathbf{p}_{\mathbf{x}^*}^*$.
		\Statex \textbf{Output}: $\mathbf{x}^*,\mathbf{p}^*$ at $t=1$.
	\end{algorithmic}
\end{algorithm}

Algorithm \ref{alg:exhuast} summarizes the procedure.
 \subsection{Convex Relaxation}
An approximate solution can be obtained by means of relaxation of $x_{\mu,\alpha}^t$. In order to relax the problem, in \cite{wang2017optimization,jin2015resource} the binary variable $x_{\mu,\alpha}^t$ is replaced with a real variable $x_{\mu,\alpha}^t$, taking values in $\left[0,1\right]$. In this paper, however, we apply a more general approach by replacing $x_{\mu,\alpha}^t$ with $h(x_{\mu,\alpha}^t)$. It is intended that the value of the objective function doesn't change with this replacement of the binary variable $x_{\mu,\alpha}^t$. Given this condition on $h(\cdot)$, the relaxed problem is equivalent to the original problem for binary values of $x_{\mu,\alpha}^t$. To have the same value for $h(x)$ and $x$ at $x=0,1$, we have
\begin{equation}
h(x)=\left\lbrace\begin{matrix}
1 & x=1,\\
0 & x=0.
\end{matrix}\right.
\end{equation}
The function $h(x)=x^a,a>0$, satisf\/ies the above condition. 

For $\beta>1$ the objective function is always negative. Hence, the maximization problem of (\ref{eq:mobilityaware}) is equivalent to the \textit{minimization} of absolute value of the objective function, which with proper choice of $a$ will be a convex one. The necessary and suff\/icient condition on $a$, assuming $\beta>1$, is given in lemma \ref{lem:convex}. The  truncated relaxed optimization problem for $\beta>1$ is given by
\begin{alignat}{2}
&\min_{x_{\mu,\alpha}^t,p_{\mu,\alpha}^t} \sum_{\mu}\sum_{\alpha}\sum_{t}\frac{h(x^t_{\mu,\alpha})}{\ \left(r_{\mu,\alpha}^t p_{\mu,\alpha}^t \right)^{\beta-1}}\label{eq:relaxedfinal}\\
&\text{s.t.}\; \,r_{\mu,\alpha}^t=\left((1-\eta_0)x_{\mu,\alpha}^{t-1}+\eta_0\right)R_{\mu,\alpha}^t \; \forall \mu\in\mathcal{U},\forall \alpha\in\mathcal{A}, t\in\mathbb{N}^T,\nonumber\\
&\;\quad\,\sum_{\mu\in\mathcal{U}}x_{\mu,\alpha}^t p_{\mu,\alpha}^t\leq 1 \qquad \forall \alpha\in\mathcal{A}, \,\,t\in\mathbb{N}^T,\nonumber\\
&\;\quad\,\sum_{\alpha\in \mathcal{A}}x_{\mu,\alpha}^t=1 \qquad\quad\ \; \forall \mu\in\mathcal{U},\,\,t\in\mathbb{N}^T,\nonumber\\
&\;\quad\, x_{\mu,\alpha}^t\in[ 0,1],p_{\mu,\alpha}^t\in\left(0,1\right]\quad\forall \alpha\in\mathcal{A},\,\mu\in\mathcal{U},\,t\in\mathbb{N}^T,\nonumber
\end{alignat}
where $\mathbb{N}^T=\lbrace 1,2,\ldots,T\rbrace$.

\begin{lemma}\label{lem:convex}
The objective function of the relaxed problem described in (\ref{eq:relaxedfinal}) with $h(x)=x^a$ is convex for $\beta>1$ iff 
$$\left\lbrace\begin{matrix}
a\geq & \beta & t=1,\\
a\geq &2\beta-1 & t>1.
\end{matrix}\right.$$
\end{lemma}	
\begin{proof}
See appendix \ref{ap:convex}.
\end{proof}

Analogous to the convexity constraints of lemma \ref{lem:convex}, which is used throughout this paper, the following lemma states the required constraints for the case of $0.5\leq\beta<1$:
\begin{lemma}\label{lem:concave}
	The relaxed objective function given by $$ h(x^t_{\mu,\alpha})\left(r_{\mu,\alpha}^t p_{\mu,\alpha}^t \right)^{1-\beta}$$ with $h(x)=x^a$ is concave for $0.5\leq\beta<1$ iff 
	$$\left\lbrace\begin{array}{ll}
	a\leq\beta & t=1,\\
	a\leq 2\beta-1 & t>1.
	\end{array}\right.$$
\end{lemma}	
\begin{proof}
	See appendix \ref{ap:concave}.
\end{proof}

The above lemmas provides many different choices for the relaxing function $h(x)$, each of which may result in a specific algorithm. Also since the summation of convex functions is also convex, $h(x)$ may be chosen as the sum of several different functions satisfying the conditions of the proved lemmas. In the remainder of this paper, $\beta>1$ is chosen with $h(x)=x^{2\beta-1}$ since, as will be shown, this choice results in  closed form solutions for optimization variables in each iteration and consequently reduces the computational cost of the algorithm. Also, this choice relaxes the objective function to a convex one for all values of $t$, including $t=1$, for which the value of $r_{\mu,\alpha}^1$ is not an optimization variable. This value is a function of $x_{\mu,\alpha}^0$, which is the result of the algorithm after being executed at the previous service time.

\subsection{Proposed Algorithm}
As discussed earlier, the objective function in (\ref{eq:relaxedfinal}) is a convex one. Consequently, we use the Lagrangian multiplier method to solve the dual problem. The Lagrangian function is as follows  \cite{boyd2004convex}:
\begin{align}
L(\mathbf{x},\mathbf{p},\mathbf{r},&\mathbf{\lambda},\mathbf{\zeta},\mathbf{\gamma})=\sum_{\mu}\sum_{\alpha}\sum_{t}\frac{(x^t_{\mu,\alpha})^{2\beta-1}}{\left(r_{\mu,\alpha}^t p_{\mu,\alpha}^t \right)^{\beta-1}}\nonumber\\
&+\sum_{\mu}\sum_{\alpha}\sum_{t}\lambda_{\alpha}^t x_{\mu,\alpha}^t p_{\mu,\alpha}^t - \sum_{\alpha}\sum_{t}\lambda_{\alpha}^t\nonumber\\
&+\sum_{\mu}\sum_{\alpha}\sum_{t}\zeta_\mu^t x_{\mu,\alpha}^t -\sum_{\mu}\sum_{t}\zeta_\mu^t\nonumber\\
&+\sum_{\mu}\sum_{\alpha}\sum_{t\geq 2}\gamma_{\mu,\alpha}^t r_{\mu,\alpha}^t-\gamma_{\mu,\alpha}^t x_{\mu,\alpha}^{t-1}(1-\eta_0)R_{\mu,\alpha}^t\nonumber\\
&-\sum_{\mu}\sum_{\alpha}\sum_{t\geq 2}\gamma_{\mu,\alpha}^t\eta_0R_{\mu,\alpha}^t,\label{eq:Lagrangian1}
\end{align}
where $x_{\mu,\alpha}^t\in[ 0,1],p_{\mu,\alpha}^t\in(0,1]$. Recall that $r_{\mu,\alpha}^1$ is not an optimization variable and hence its corresponding multiplier, $\gamma_{\mu,\alpha}^1$, is always zero and the last sum is over $2\leq t\leq T$.

The optimal solution to the problem can be found by solving the following dual problem:
\begin{equation}\label{eq:LagrangianDual}
\begin{matrix}
\text{max}&\min & L(\mathbf{x},\mathbf{p},\mathbf{r},\mathbf{\lambda},\mathbf{\zeta},\mathbf{\gamma})\\
\mathbf{\lambda},\mathbf{\zeta},\mathbf{\gamma}&\mathbf{x},\mathbf{p},\mathbf{r}\\
\text{s.t. }&\mathbf{\lambda}\succeq 0 &
\end{matrix}
\end{equation}
This optimization will be solved in two steps. At the f\/irst step, the minimization problem should be solved with known $\mathbf{\lambda},\mathbf{\zeta},\mathbf{\gamma}$. To solve the minimization, KKT conditions may be applied. For $t\geq 2$, setting derivatives of the Lagrangian with respect to variables equal to zero yields
\begin{align}
\frac{\partial L}{\partial x_{\mu,\alpha}^t}&=\frac{(2\beta-1)(x_{\mu,\alpha}^t)^{2\beta-2}}{(r_{\mu,\alpha}^t)^{\beta-1}(p_{\mu,\alpha}^t)^{\beta-1}}
+\lambda_{\alpha}^t p_{\mu,\alpha}^t+\zeta_\mu^t\nonumber\\
&-\gamma_{\mu,\alpha}^{t+1}(1-\eta_0)R_{\mu,\alpha}^{t+1}=0,\label{eq:derivX}\\
\frac{\partial L}{\partial p_{\mu,\alpha}^t}&=-\frac{(\beta-1)(x_{\mu,\alpha}^t)^{2\beta-1}}{(r_{\mu,\alpha}^t)^{\beta-1}(p_{\mu,\alpha}^t)^{\beta}}+\lambda_{\alpha}^tx_{\mu,\alpha}^t=0,\label{eq:derivP}\\
\frac{\partial L}{\partial r_{\mu,\alpha}^t}&=-\frac{(\beta-1)(x_{\mu,\alpha}^t)^{2\beta-1}}{(r_{\mu,\alpha}^t)^{\beta}(p_{\mu,\alpha}^t)^{\beta-1}}+\gamma_{\mu,\alpha}^t=0.\label{eq:derivR}
\end{align}
Note that $r_{\mu,\alpha}^1$ has a predef\/ined value based on the AP allocation at the previous service time and is not an optimization parameter. Thus, for $t=1$ we have
\begin{align}
\frac{\partial L}{\partial x_{\mu,\alpha}^1}&=\frac{(2\beta-1)(x_{\mu,\alpha}^1)^{2\beta-2}}{(r_{\mu,\alpha}^1)^{\beta-1}(p_{\mu,\alpha}^1)^{\beta-1}}
+\lambda_{\alpha}^1 p_{\mu,\alpha}^1+\zeta_\mu^1\nonumber\\
&-\gamma_{\mu,\alpha}^{2}(1-\eta_0)R_{\mu,\alpha}^{2}=0,\label{eq:derivX1}\\
\frac{\partial L}{\partial p_{\mu,\alpha}^1}&=-\frac{(\beta-1)(x_{\mu,\alpha}^1)^{2\beta-1}}{(r_{\mu,\alpha}^1)^{\beta-1}(p_{\mu,\alpha}^1)^{\beta}}+\lambda_{\alpha}^1x_{\mu,\alpha}^1=0.\label{eq:derivP1}
\end{align}
By using (\ref{eq:derivP}) and (\ref{eq:derivR}) we obtain
\begin{equation}
\lambda_{\alpha}^t x_{\mu,\alpha}^{t^*} p_{\mu,\alpha}^{t^*}=\gamma_{\mu,\alpha}^t r_{\mu,\alpha}^{t^*}.\label{eq:xpr}
\end{equation}
The optimum values of $x_{\mu,\alpha}^t$ and $r_{\mu,\alpha}^t$ can be expressed as a function of $p_{\mu,\alpha}^t$ as follows:
\begin{align}
r_{\mu,\alpha}^{t^*}&=\frac{(\lambda_{\alpha}^t)^{\frac{2\beta-1}{\beta-1}}}{(\beta-1)^{\frac{1}{\beta-1}}(\gamma_{\mu,\alpha}^t)^2}(p_{\mu,\alpha}^{t^*})^{\frac{3\beta-2}{\beta-1}}\qquad t\geq 2\label{eq:rOptp},\\
x_{\mu,\alpha}^{t^*}&=\left(\frac{(r_{\mu,\alpha}^t)^{\beta-1}(p_{\mu,\alpha}^t)^\beta \lambda_{\alpha}^t}{\beta-1}\right)^\frac{1}{2\beta-2}.\label{eq:xOptp}
\end{align}

By using (\ref{eq:derivX}), (\ref{eq:xOptp}), and (\ref{eq:rOptp}), we obtain
\begin{align}
&\frac{(2\beta-1)\dfrac{(\lambda_{\alpha}^t)^{\frac{\beta(2\beta-2)}{\beta-1}}}{(\beta-1)^{\frac{2\beta-2}{\beta-1}}(\gamma_{\mu,\alpha}^t)^{2\beta-2}}(p_{\mu,\alpha}^{t^*})^{\frac{(2\beta-1)(2\beta-2)}{\beta-1}}}{\dfrac{(\lambda_{\alpha}^t)^{2\beta-1}}{(\beta-1)(\gamma_{\mu,\alpha}^t)^{2\beta-2}}(p_{\mu,\alpha}^{t^*})^{4\beta-3}}\nonumber\\
&+\lambda_{\alpha}^t p_{\mu,\alpha}^{t^*}+\zeta_\mu^t -\gamma_{\mu,\alpha}^{t+1}(1-\eta_0)R_{\mu,\alpha}^{t+1}=0.
\end{align}
The optimum value of $p_{\mu,\alpha}^t$ is then given by
\begin{equation}\label{eq:pOpt}
p_{\mu,\alpha}^{t^*}=\frac{\beta-1}{3\beta-2}\frac{\gamma_{\mu,\alpha}^{t+1}(1-\eta_0)R_{\mu,\alpha}^{t+1}-\zeta_\mu^t}{\lambda_{\alpha}^t}.
\end{equation}
Note that we have $\gamma_{\mu,\alpha}^{T+1}=0$.

Equation (\ref{eq:pOpt}) shows that the solution's dependency  at service time $t$ on the future is by means of the Lagrangian multiplier associated with the constraint on rates, i.e., $\gamma_{\mu,\alpha}^{t+1}$. Hence, the mobility-unaware scheme, which is achieved by setting $T=1$, doesn't have this coeff\/icient.   

Up to now, the minimization in (\ref{eq:LagrangianDual}) is solved for given values of the Lagrangian multipliers ($\mathbf{\lambda},\mathbf{\zeta},\mathbf{\gamma} $). In the second step, the maximization in (\ref{eq:LagrangianDual}) is solved by using the obtained values for $\lbrace \mathbf{x}^*,\mathbf{p}^*,\mathbf{r}^* \rbrace$:
\begin{equation}\label{eq:step2dual1}
\max\limits_{\mathbf{\lambda}\succeq 0,\mathbf{\zeta},\mathbf{\gamma}}\quad L(\mathbf{x}^*,\mathbf{p}^*,\mathbf{r}^*,\mathbf{\lambda},\mathbf{\zeta},\mathbf{\gamma}).
\end{equation}
This can be done iteratively using the gradient ascent method. 

The derivatives of $L$ with respect to the Lagrangian multipliers are computed as follows:
\begin{align}
\frac{\partial L(\mathbf{x}^*,\mathbf{p}^*,\mathbf{r}^*,\mathbf{\lambda},\mathbf{\zeta},\mathbf{\gamma})}{\partial \lambda_{\alpha}^t} &=\sum_{\mu}x_{\mu,\alpha}^{{t}^*}p_{\mu,\alpha}^{{t}^*}-1,\\
\frac{\partial L(\mathbf{x}^*,\mathbf{p}^*,\mathbf{r}^*,\mathbf{\lambda},\mathbf{\zeta},\mathbf{\gamma})}{\partial \zeta_{\mu}^t} &=\sum_{\alpha}x_{\mu,\alpha}^{t^*}-1,\\
\text{{\small $\frac{\partial L(\mathbf{x}^*,\mathbf{p}^*,\mathbf{r}^*,\mathbf{\lambda},\mathbf{\zeta},\mathbf{\gamma})}{\partial \gamma_{\mu,\alpha}^t}$}} &\text{{\small $=r_{\mu,\alpha}^{t^*}-x_{\mu,\alpha}^{t-1^*}(1-\eta_0)R_{\mu,\alpha}^t-\eta_0R_{\mu,\alpha}^t$}}.
\end{align}

The Lagrangian multipliers are then updated in the ascent direction of the gradient iteratively as follows:
\begin{align}
\lambda_{\alpha}^t(n+1) &=\left[\lambda_{\alpha}^t(n)+ \epsilon_\lambda(n)\left(\sum_{\mu}x_{\mu,\alpha}^{{t}^*}p_{\mu,\alpha}^{{t}^*}-1\right)\right]^+,\label{eq:gradient1}\\
\zeta_{\mu}^t(n+1) &=\lambda_{\alpha}^t(n)+ \epsilon_\zeta(n)\left(\sum_{\alpha}x_{\mu,\alpha}^{t^*}-1\right),\label{eq:gradient2}\\
\text{\footnotesize $\gamma_{\mu,\alpha}^t(n+1)$} &=\text{\footnotesize $ \gamma_{\mu,\alpha}^t(n)+ \epsilon_\gamma(n)\left(r_{\mu,\alpha}^{t^*}-(x_{\mu,\alpha}^{t-1^*}(1-\eta_0)+\eta_0)R_{\mu,\alpha}^t\right)$},\label{eq:gradient3}
\end{align}
where $\epsilon_\lambda(n),\epsilon_\zeta(n),\epsilon_\gamma(n)$ are the step sizes indicating the amount of movement in the ascent direction of the gradient with $[\cdot]^+=\max\lbrace.,0\rbrace$. Choosing small step sizes leads to a slow convergence. If the step size is chosen too large, the algorithm might diverge. We have obtained convergence by setting $\epsilon_\lambda(n)=\epsilon_\zeta(n)=100\epsilon_\gamma(n)=\epsilon(n)$. The parameter $\epsilon(n)$ may be constant or decay similarly to \cite{jin2015resource}.

Finally, each user should connect to only one AP at each service time. Hence, the relaxed variable will be recovered by assigning each user to the AP with the highest value of $x_{\mu,\alpha}^{t^*}$.
\begin{equation}\label{eq:Xopt1}
x_{\mu,\alpha}^{t^*}\gets\left\lbrace \begin{matrix}
1, & x_{\mu,\alpha}^{t^*}=\arg\max_{\alpha\in\mathcal{A}} x_{\mu,\alpha}^{t^*}\\
0, & \text{otherwise}.
\end{matrix}\right.
\end{equation}
Note that the only value needed for the implementation is $x_{\mu,\alpha}^{1^*}$ and $p_{\mu,\alpha}^{1^*}$.
After converting $x_{\mu,\alpha}^{1^*}$ to a binary variable, $p_{\mu,\alpha}^{1^*}$ should be normalized in order to use all of the allocated resources:
\begin{equation}\label{eq:Popt1}
p_{\mu,\alpha}^{1^*}\gets \frac{p_{\mu,\alpha}^{1^*}}{x_{\mu,\alpha}^{1^*}\sum_{\mu}p_{\mu,\alpha}^{1^*}}.
\end{equation}
 
We propose the so-called \textit{Mobilty-aware VLC Resource (MVR)} allocation algorithm, whose pseudo-code is given by Algorithm \ref{alg:lagrange}.
\begin{algorithm}
	\caption{MVR}\label{alg:lagrange}
	\begin{algorithmic}[1]
		\Statex \textbf{Input}
		\Statex $T$: The total number of future service times influencing current service time.
		\Statex $y_{\mu}^t$: Position of user $\mu$ at service time $0\leq t\leq T$.
		\Statex \textbf{Initialization}
		\Statex $n \gets 1$.
		\Statex Lagrangian multipliers $\lambda_{\alpha}^t(0),\zeta_{\mu}^t(0),\gamma_{\mu,\alpha}^t(0)$.
		\Statex Step size: $\epsilon_\lambda(n),\epsilon_\zeta(n),\epsilon_\gamma(n)$.
		\ForAll {$\mu\in\mathcal{U}$,$\alpha\in\mathcal{A}$,$1\leq t\leq T$}
		\State CU calculates $R^t_{\mu,\alpha}$.
		\EndFor
		\While {$n<N$}
		\ForAll {$\mu\in\mathcal{U}$,$\alpha\in\mathcal{A}$,$0\leq t\leq T$}
		\State Find optimal $\mathbf{p}^*$,$\mathbf{x}^*$,$\mathbf{r}^*$ using (\ref{eq:pOpt}),(\ref{eq:xOptp}) and (\ref{eq:rOptp}).
		\State Update multipliers according to (\ref{eq:gradient1}),(\ref{eq:gradient2}) and (\ref{eq:gradient3}). 
		\EndFor
		\State $n \gets n+1$.
		\EndWhile
		\State Recover the relaxed variables, $\mathbf{x}^*$, according to (\ref{eq:Xopt1}).
		\State Normalize $\mathbf{p}^*$ based on obtained  $\mathbf{x}^*$ using (\ref{eq:Popt1}).
		\Statex \textbf{Output}: $\mathbf{x}^*,\mathbf{p}^*$ at $t=1$.
	\end{algorithmic}
\end{algorithm}

To evaluate the long-term performance of the proposed algorithm, we use the \textit{total objective function}, which is the sum of the objective function over all service times. Fig. \ref{fig:minlp} shows that the performance of the proposed algorithm (MVR) is close to the one obtained via exhaustive search in a 3-user environment\footnote{Since the exhaustive search is very time-consuming for large number of users, only 3 users are simulated here.} with 2 APs. The results are compared for $T=1$ with the heuristic algorithm of  \cite{wang2017optimization}, called Joint Optimization Algorithm (JOA), for the mobility unaware scheme. This f\/igure shows that the relaxed algorithm performs close to the optimum of the MINLP problem.

\begin{figure}[t]
	\centering
	\includegraphics[width=0.4\textwidth]{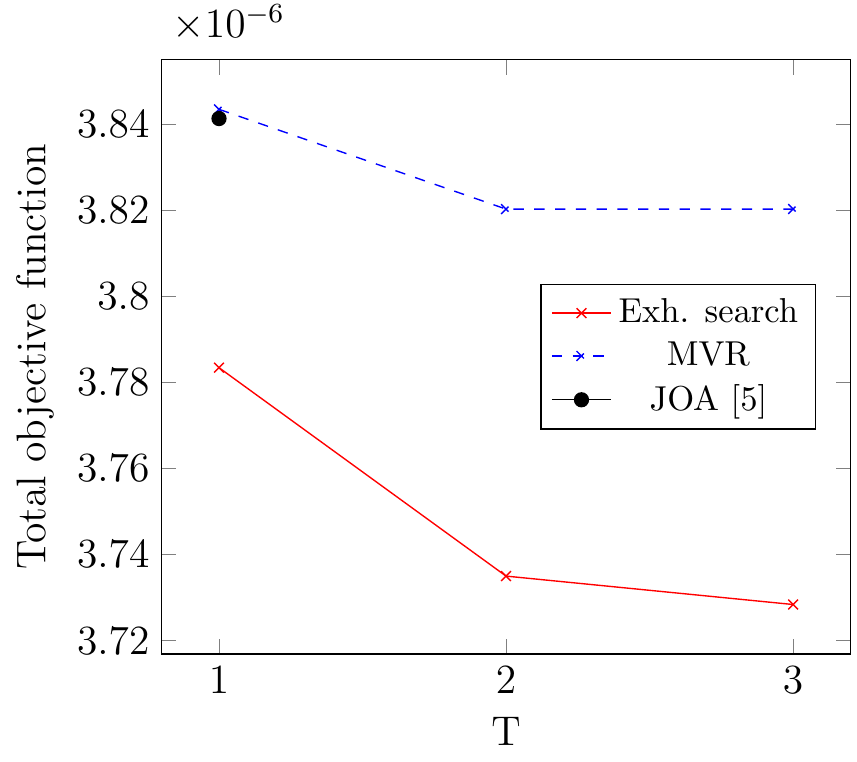}
	\caption{Comparison of optimum performance of MVR, JOA \cite{wang2017optimization}, and exhaustive search (system is run for one minute with $\beta=2$ and $\eta_0=0.75$.)}
	\label{fig:minlp}
\end{figure}

\section{Mobility Model}\label{sec:mobility}
To evaluate the performance of network algorithms in the presence of mobile users, a proper model should be used to describe the mobility of users. Various mobility models can be used, amongst them are Brownian motion, random walk, Markovian models \cite{hong1999group}, region splitting \cite{ashtiani2003mobility}, and Levy walk \cite{rhee2011levy}. Further detailed information about mobility models can be found in \cite{hong1999group,camp2002survey}. In random walk mobility models, each user moves along a straight line at a constant speed. One of the most widely used random walk models is the random waypoint mobility model. This model has been used in the simulation of networks and has been studied extensively \cite{bettstetter2003node,bettstetter2004stochastic}. In this section, we describe the method of predicting users' mobility in greater detail. An upper-bound on service time is also derived in this section.

\subsection{Prediction}
As discussed earlier, for $T>1$ the proposed algorithm requires the prediction of the future data rates experienced by users. Therefore, the need  for predicting the users' future locations is inevitable. It has been previously shown that human mobility exhibits a high potential predictability \cite{song2010limits}. Hence, an approximately accurate prediction of the users' future locations is possible based on their mobility patterns in the past.

In order to predict the users' position, in random walk mobility models it is enough for the positioning system to measure the users' positions at two consecutive times. With the information obtained from these measurements, the CU can calculate the velocity of each user. Since the movement in these models occurs at constant speed, the future path of the users can be predicted by using both velocity and current position of the users. Note that the prediction in these models can lead to a wrong value of future position, if one user stops and changes direction sometime in the intended future. 

For simplicity, let us assume that the users exactly follow a random waypoint mobility model. However, if the users' mobility show some small deviations from the constant speed model, these deviations may be included as noise in the mobility model of the users and a Kalman f\/ilter may be used to f\/ilter out the noise \cite{xiong2016cooperative}. Another approach to prediction, which is more general, may be achieved with the aid of multistep-ahead time series prediction \cite{cheng2006multistep}. 

\subsection{Service Time in Random Waypoint Model}
In random waypoint mobility model, each user chooses a point in the room randomly with uniform distribution and then moves at constant speed from its current position to the chosen point. There the user stops and waits for a random time and then chooses another location with an optionally another speed. This moving process repeats itself again and again.

According to the def\/inition of the random waypoint model, stop points are chosen uniformly. As a result, the distribution of stop points conditioned on knowing the current position and direction of movement will also be uniform.  Assuming that each service time lasts $\tau_p$ seconds, the probability that a user stops or changes direction before $T$  service times, which is also the probability of mis-prediction during $T$ service times, is given by
\begin{align}
&\Pr(\text{stop before } \vec{y}_\mu^{\,T}|\vec{y}_\mu^{\,0} ,\vec{v}_\mu^{\,0})=\nonumber\\
&\Pr(l_s\leq l(T\tau_p)|\vec{y}_\mu^{\,0},\vec{v}_\mu^{\,0})=\frac{v_0T\tau_p}{l(\theta_0,\vec{y}_\mu^{\,0} )},
\end{align}
where $\vec{v}_\mu^{\,0}=v_0\angle \theta_0$ is the speed vector of the user, $\theta_{0}$ is the angle of speed vector which shows the direction of movement, and $l(\theta_0,\vec{y}_\mu^{\,0} )$ is the random variable of remaining travel length of user before its pause conditioned on the current position, $\vec{y}_\mu^{\,0}$, and direction of movement. 
 
Clearly, the probability of mis-prediction should be small during each run of the algorithm, which considers $T$ later service times. This fact can be used to obtain an upper bound on the service time interval. Thus, $\tau_p$  is chosen such that the expected value of this probability be lower than a certain value, say, $\delta>0$, as follows:
\begin{align}
 &\mathbb{E}[\Pr(\text{stop before } \vec{y}_\mu^{\,T}|\vec{y}_\mu^{\,0} ,\vec{v}_\mu^{\,0})]\leq \delta\Rightarrow\nonumber\\
 &\tau_p\leq\mathbb{E}\left[\frac{\delta l(\theta_0,\vec{y}_\mu^{\,0} )}{v_0T}\right].\label{eq:prdelta}
 \end{align}
According to \cite{bettstetter2004stochastic}, if speed is uniformly distributed in $[v_\text{max},v_\text{min}]$, $v_\text{min}>0$\footnote{$v_\text{min}=0$ is equivalent to a pause}, we have
\begin{align}
\mathbb{E}\left[\frac{l(\theta_0,\vec{y}_\mu^{\,0} )}{v_0}\right]&= \mathbb{E}[l(\theta_0,\vec{y}_\mu^{\,0} )]\int_{v_\text{min}}^{v_\text{max}}\frac{1}{v}f_V(v)dv\nonumber\\ 
&= \frac{\ln(v_\text{max}/v_\text{min})}{v_\text{max}-v_\text{min}}\mathbb{E}[l(\theta_0,\vec{y}_\mu^{\,0} )],
\end{align}
whereby $\mathbb{E}[l(\theta_0,\vec{y}_\mu^{\,0})]$ was calculated for different environments in \cite{bettstetter2004stochastic}. Hence, the upper bound on service time for a required performance may be found in terms of the environment parameters and $T$ as follows:
\begin{equation}
\tau_p\leq \frac{\delta}{T}\frac{\ln(v_\text{max}/v_\text{min} )}{v_\text{max}-v_\text{min}}\mathbb{E}[l] .
\end{equation}

\section{Numerical Results}
\label{sec:numerical}
In this section, we verify the proposed method using various simulations. The mobility aware algorithms in \cite{li2016mobility} and \cite{zhang2018anticipatory} focus on rate maximization without fairness. Hence, comparing JOA or MVR with them is not fair and the results will be trivial. The approach of \cite{elazzouni2018qos} is limited to a single AP hybrid scenario which is not applicable to the multiple AP scenario of this paper. Hence, the performance of the proposed algorithm is compared with the JOA algorithm \cite{wang2017optimization} which uses a similar utility function but without considering mobility-awareness. 

Although the proposed algorithm is applicable to arbitrary number of APs, we evaluate our proposed method in presence of minimal equipments. Hence the main simulation scenario under consideration is an 8m $\times$ 4m indoor room equipped with 2 VLC APs. Users move according to the random waypoint mobility model described in the previous section. The speed of users is uniformly distributed between 0 m/s and 1 m/s with a pause time between 0 s and 1 s. The stationary distribution of the random waypoint mobility model \cite{navidi2004stationary} is used as initial distribution of the users. Frequency reuse is implemented for the APs in order to eliminate interference between adjacent cells. The remaining simulation parameters are given in Table \ref{tab:parameters} in compliance with \cite{wang2017optimization,wang2015dynamic,jin2015resource}. All simulations are run for 15 minutes, which is equivalent to 3000 realizations of service times with $\beta=2$ and $\eta_0=0.75$.

\begin{table}[bt]
	\caption{Simulation Parameters}
	\centering
	\begin{tabular}{ll}
		\hline
		Parameter & Value\\
		\hline
		Vertical distance between APs and users, $h$ & 2.3 m\\
		Average transmitted optical power $P$ & 10 W\\
		Baseband modulation bandwidth $B$ & 20 Mhz\\
		Physical area of a PD, $A$ & 1 cm$^2$\\
		Half-intensity radiation angle $\Phi_{1/2}$ & 30 deg.\\
		Gain of optical f\/ilter, $T_s (\psi)$ & 1.0 \\
		Receiver FoV semi-angle, $\Psi_f$  & 90 deg.\\
		Refractive index, $n$ & 1.5\\
		Optical to electric conversion eff\/iciency, $\kappa$ & 0.53\\
		Noise power spectral density $N$ & {\scriptsize 1e-19 W/Hz}\\
		Resource allocation service-time $\tau_p$ & 300ms\\
		\hline
	\end{tabular}\label{tab:parameters}
\end{table}

\subsubsection{Throughput}
The throughput of the system is plotted in Fig. \ref{fig:rate} for different prediction levels in comparison with the so-called JOA algorithm proposed by Wang \textit{et al.} \cite{wang2017optimization}. The results show that the system throughput of the proposed method increases with the increase of prediction level. We also observe that MVR gives throughputs higher than that of JOA. 

Because of user diversity, as shown in Fig. \ref{fig:rate}, throughput of the system is constantly increasing when the number of users is smaller than 20. However as the number of users increase further, the number of handovers increases and the negative impact of handovers will surpass the positive effect of user diversity decreasing the throughput. 
\begin{figure}[bt]
	\centering
	\includegraphics[width=0.45\textwidth]{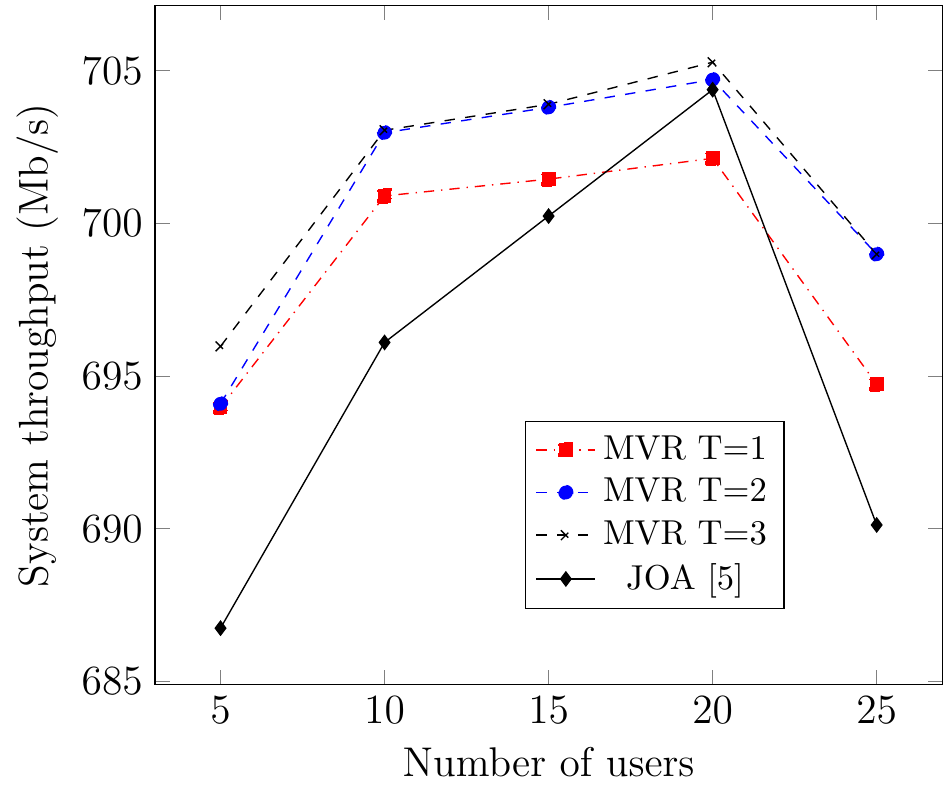}
	\caption{Impact of prediction on system throughput (2 APs).}
	\label{fig:rate}
\end{figure}

Similar results may be obtained by increasing number of the APs. For example Fig. \ref{fig:rate4} shows the average throughput of the proposed algorithm in presence of 4 APs in an 8m$\times$8m room. In this case the gain of the algorithm for $T>2$ is negligible.

\begin{figure}[bt]
	\centering
	\includegraphics[width=0.45\textwidth]{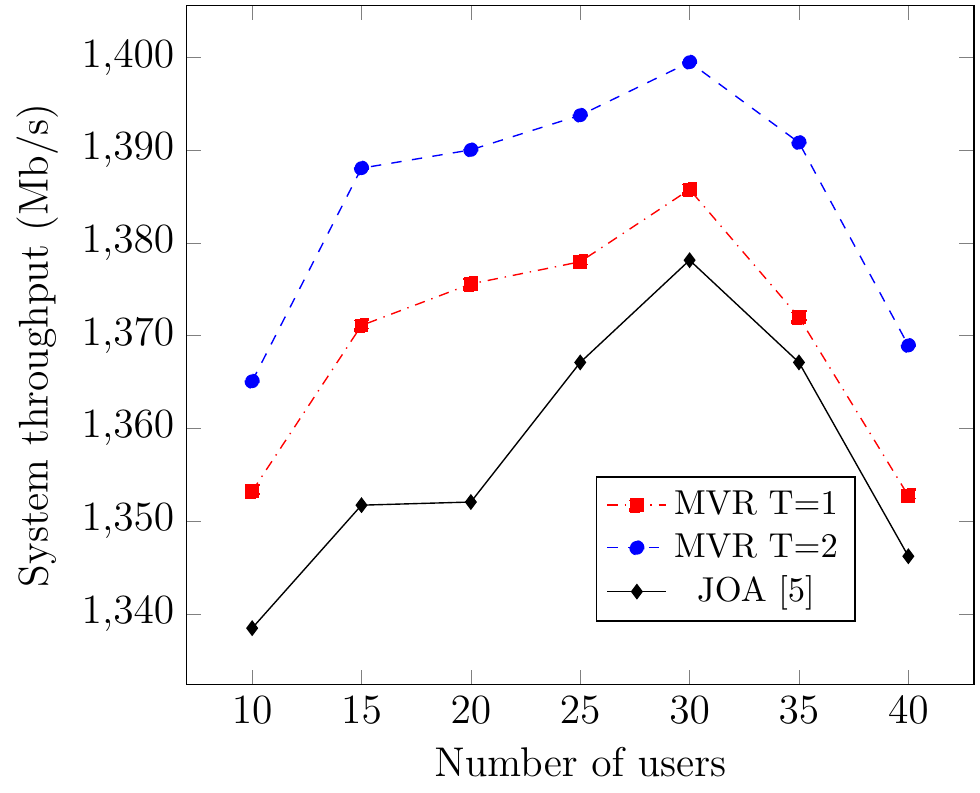}
	\caption{Average system throughput in presence of 4 APs.}
	\label{fig:rate4}
\end{figure}

\subsubsection{Saturation effect}
Fig. \ref{fig:objective} shows the total objective function of the system with different levels of prediction. The results show that for $T>3$ the gain of the algorithm becomes negligible.

\begin{figure}[bt]
	\centering
	\includegraphics[width=0.45\textwidth]{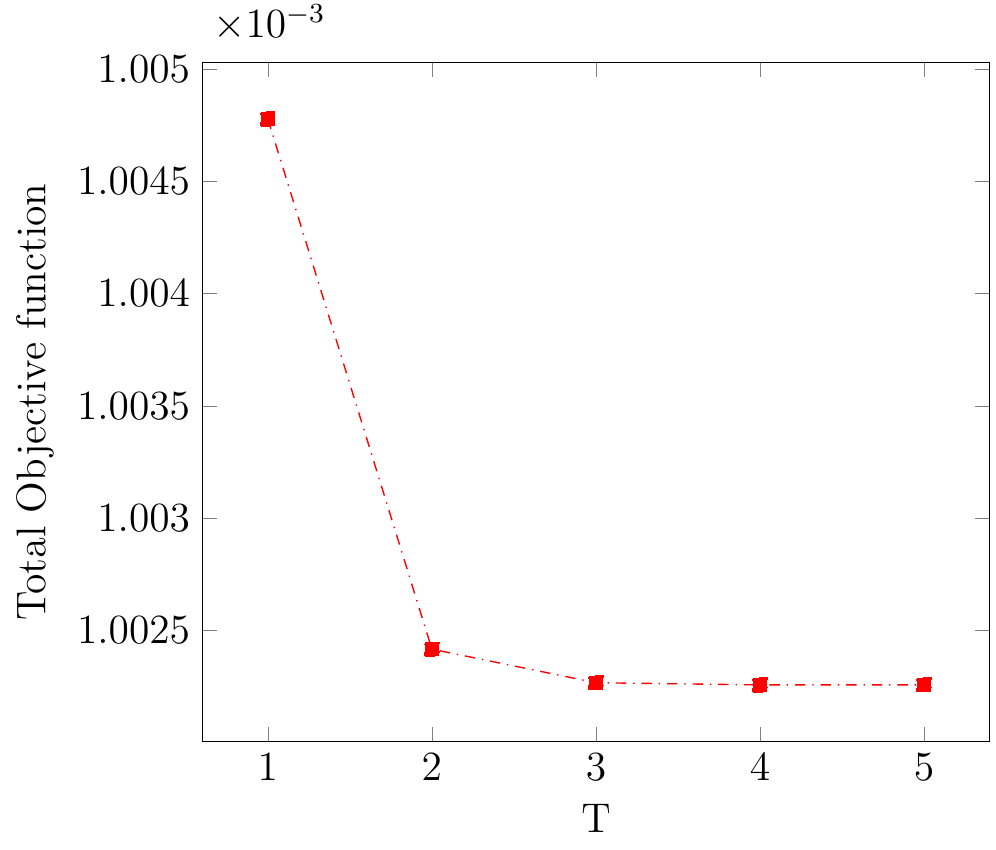}
	\caption{Total objective function of MVR vs. prediction level for 15 users.}\label{fig:objective}
\end{figure}
\begin{figure}[bt]
\centering
\includegraphics[width=0.45\textwidth]{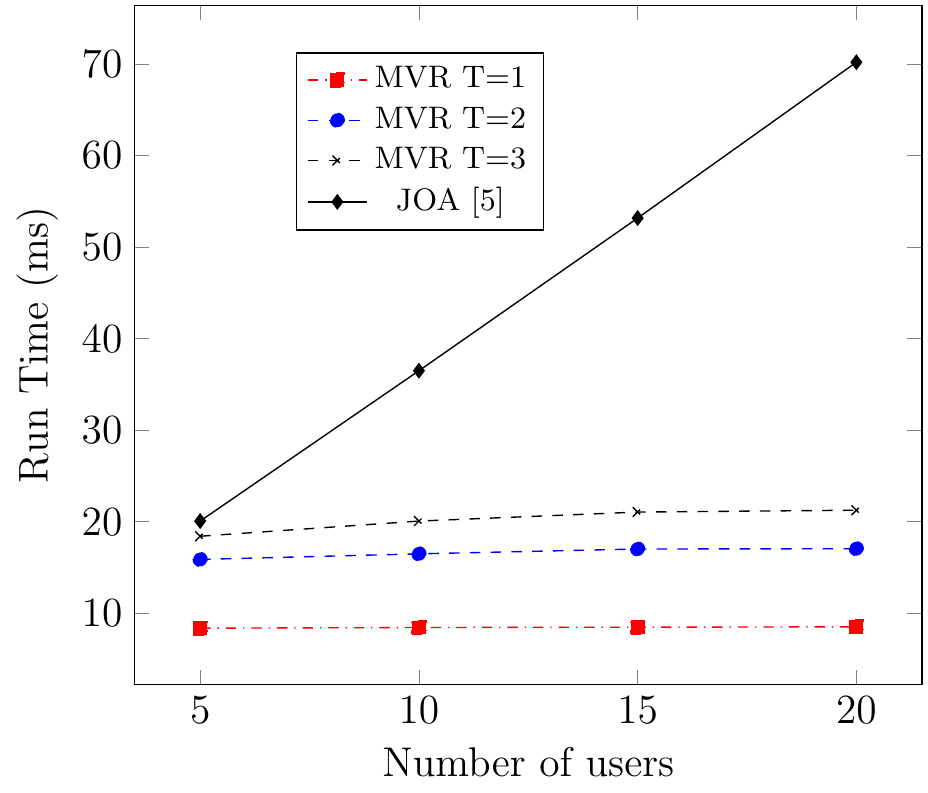}
\caption{Average run-time of algorithm for each service time.}\label{fig:runtime}
\end{figure}

\subsubsection{Runtime}
The algorithm was implemented on a system with a single 2.16 GHz Intel$^{\circledR}$ Celeron$^{\circledR}$ dual-core CPU with 4.00 GB RAM. Fig. \ref{fig:runtime} depicts the average runtime of the different algorithms for one service time. We observe that the runtime increases with the number of users and prediction levels. However, this increase is higher for the algorithm of \cite{wang2017optimization} due to the need for f\/inding argmax in each iteration. This plot conf\/irms that MVR is fast enough to be implemented in VLC networks, being faster than JOA especially for high number of users. 

While we weren't able to achieve convergence for a mobility-aware algorithm based on the decomposition technique of \cite{jin2015resource}, it still might be possible to add mobility prediction into the decomposition scheme of JOA which can be a topic for future research. However, the run-time of such algorithm will grow exponentially with the increase of prediction level, $T$. The decomposition approach of JOA requires a search for each user in each iteration between $|\mathcal{A}|$ values. Hence the total number of searches will be $|\mathcal{U}|\times|\mathcal{A}|$ in each iteration. Assuming $I$ iterations is required for convergence of the algorithm, a mobility-aware decomposition approach will require $|\mathcal{U}|\times|\mathcal{A}|^T\times I$ searches since for each user the parameters should be found for $T$ service-times. Hence, such an algorithm would not be suitable for online purposes. In contrast, MVR only needs to search after convergence of the relaxed problem and it is needed to perform the search only for current service-time. This results in $|\mathcal{U}|\times|\mathcal{A}|$ overall searches, which do not exponentially grow with the prediction level and is significantly smaller than that of JOA.

\subsubsection{Handover cost}
The total number of handovers is shown in Fig. \ref{fig:handover2}. We note that increasing the prediction level leads to a slight increase in the number of handovers. This result shows that the mobility-aware optimization of the fairness function does not decrease the number of handovers, as opposed to the mobility-aware rate maximization \cite{li2016mobility}. Furthermore, the handover cost of the proposed method is lower than that of \cite{wang2017optimization}, especially for small numbers of users, for which we have observed an oscillatory behavior in the AP assignment of the JOA. In JOA, the relaxed variable is recovered by some sort of brute-force  in each iteration of the algorithm, which might result in such oscillatory behavior. In MVR, however, the algorithm solves an optimization problem with a convex objective function in which the relaxed variable changes smoothly over time and the recovery of the relaxed variable is done after the convergence of algorithm. This smooth nature of the relaxed optimization eliminates such oscillatory behavior.

\begin{figure}[bt]
	\centering
	\includegraphics[width=0.45\textwidth]{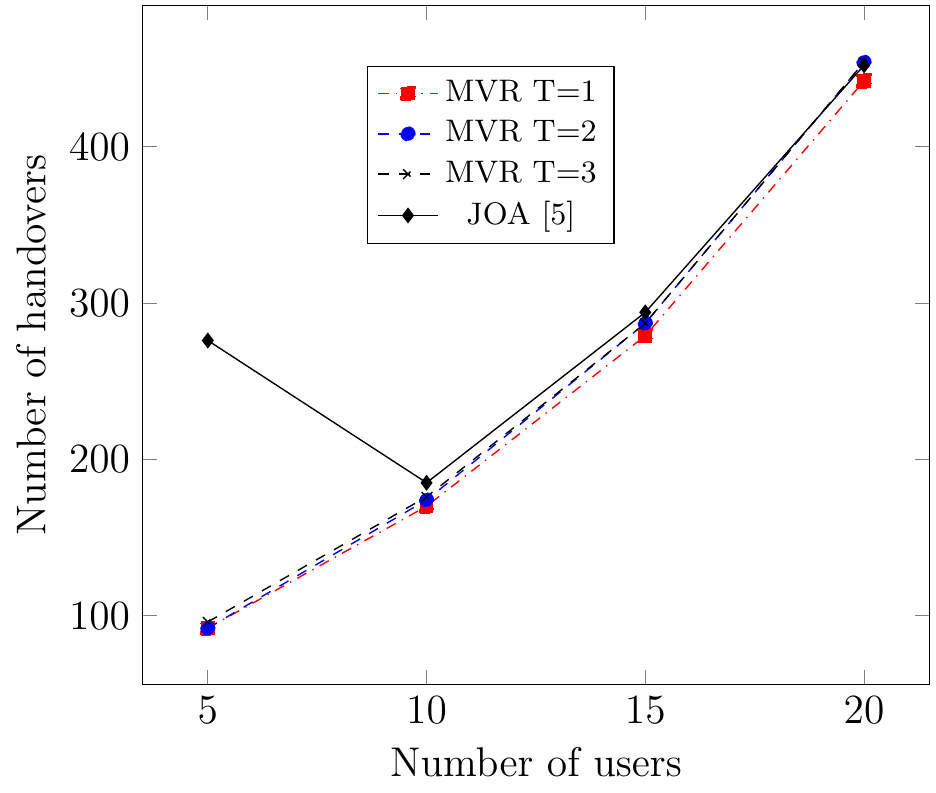}
	\caption{Total number of handovers for different algorithms.}\label{fig:handover2}
\end{figure}

\section{Conclusions}
\label{sec:conclusion}
In this paper, the awareness of user mobility was used in order to optimize the overall performance of VLC networks. By assuming the availability of a positioning system and exploiting the stationary characteristics of the indoor VLC channel, the future data rate of users was calculated and fed into the optimization problem. A novel relaxation method was proposed, which eff\/iciently solves both the mobility-aware and mobility-unaware problems. The reported simulation results show the effectiveness of our proposed method, which also eliminates the need for applying extra argmax in each iteration, as required by \cite{wang2017optimization,jin2015resource}. Hence, our proposed method is signif\/icantly faster. Our approach may be used in problems with proportional fairness objectives in order to achieve eff\/icient algorithms.


%

\appendices

\section{Proof of Lemma \ref{lem:convex}\\A Convex 3-variable Monomial}\label{ap:convex}
In this appendix some constraints on $a,b,c$ are derived such that the following function be convex:
\begin{equation}\label{Aeq:func}
f(x,y,z)=x^ay^bz^c,
\end{equation}
in which $x,y,z\geq 0$.

In order for $f(x,y,z)$ to be convex, its Hessian matrix should be positive semi-def\/inite \cite[p.~71]{boyd2004convex}. According to \cite[p.~566]{meyer2000matrix} a necessary and suff\/icient condition is that all of the principal minors of Hessian matrix should be non-negative. The Hessian matrix can be written as follows:
 \begin{align}
&\qquad\mathcal{H}=\begin{bmatrix}
\frac{\partial^2 f}{\partial^2 x} & \frac{\partial^2 f}{\partial x\partial y}  & \frac{\partial^2 f}{\partial x\partial z} \\
\frac{\partial^2 f}{\partial y\partial x} & \frac{\partial^2 f}{\partial^2 y} & \frac{\partial^2 f}{\partial y\partial z} \\
 \frac{\partial^2 f}{\partial z \partial x} & \frac{\partial^2 f}{\partial z \partial y}  &\frac{\partial^2 f}{\partial^2 z} 
\end{bmatrix}=\\&\begin{bmatrix}
a(a-1)x^{a-2}y^bz^c & abx^{a-1}y^{b-1}z^c & acx^{a-1}y^bz^{c-1}\\
abx^{a-1}y^{b-1}z^c & b(b-1)x^ay^{b-2}z^c & bcx^ay^{b-1}z^{c-1}\\
acx^{a-1}y^bz^{c-1} & bcx^ay^{b-1}z^{c-1} & c(c-1)x^ay^bz^{c-2}
\end{bmatrix}\nonumber
\end{align} 

First order principal minors give
\begin{equation}
\left\lbrace\begin{array}{ll}
a(a-1)\geq 0\Leftrightarrow & a\leq 0\quad \text{or}\quad
a\geq 1,\\
b(b-1)\geq 0\Leftrightarrow & b\leq 0 \quad \text{or}\quad
b\geq 1,
\\
c(c-1)\geq 0\Leftrightarrow & c\leq 0 \quad \text{or}\quad
c\geq 1.
\end{array}\right.
\end{equation}

Second order principal minors
\begin{equation}
\left\lbrace \begin{matrix}
ab(a+b-1)\leq 0,\\
ac(a+c-1)\leq 0,\\
bc(b+c-1)\leq 0.
\end{matrix}\right.
\end{equation}

Third order principal minor is the determinant of matrix:
\begin{align}
|\mathcal{H}|\geq 0 & \Leftrightarrow abc \begin{vmatrix}
(a-1) & a & a\\
b & (b-1) & b\\
c & c & (c-1)
\end{vmatrix}\geq 0\nonumber\\
&\Leftrightarrow abc(a+b+c-1)\geq 0.
\end{align}
Taking $b,c\leq 0$ and $a\geq 0$ implies:
\begin{equation}\label{Aeq:constraints}
\left\lbrace \begin{array}{l}
a\geq 1-b\geq 1,\\
a\geq 1-c\geq 1,\\
b+c\leq 0 \leq 1,\\
a\geq 1-b-c\geq \max{\lbrace 1-b,1-c\rbrace}.
\end{array}\right.
\end{equation}

In case of $t=1$ we have $c=0,b=1-\beta<0$. Thus the inequalities of (\ref{Aeq:constraints}) may hold simultaneously iff
\begin{equation}
a\geq 1-b=\beta.
\end{equation}

In case of $t>1$ we have $c=b=1-\beta<0$. The inequalities of (\ref{Aeq:constraints}) may hold iff
\begin{equation}
a\geq 1-2b=2\beta-1.
\end{equation}
which completes the proof.

A common solution to the inequalities in (\ref{Aeq:constraints}) for all integer $t$s can be obtained by taking $a=1-2b\geq 1-b$.  Therefore the following functions are convex when $\beta>1$\footnote{Although if $b$ and $c$ are zero together then the function will be convex, this may not happen in the proportional fairness function.}:
\begin{equation}\label{Aeq:convexgen}
f(x,y,z)=\frac{x^{2\beta-1}}{y^{\beta-1}z^{\beta-1}},\qquad f(x,y)=\frac{x^{2\beta-1}}{y^{\beta-1}}.
\end{equation}

\section{Proof of Lemma \ref{lem:concave}\\A Concave 3-variable Monomial}\label{ap:concave}
An approach similar to appendix \ref{ap:convex} can be taken for the function of (\ref{Aeq:func}) to be concave. In this case as a necessary and suff\/icient condition, the Hessian matrix should be negative semi-def\/inite.

First and third order principal minors should be non-positive and second order principal minors should be non-negative. Taking $0\leq a,b,c\leq 1$ implies:
\begin{equation}\label{Aeq:constraints2}
\left\lbrace \begin{array}{l}
b+c\leq 1,\\
a \leq 1-b-c \leq \min{\lbrace 1-b,1-c\rbrace}.
\end{array}\right.
\end{equation}

In case of $t=1$ we have $c=0,b=1-\beta$. Thus the inequalities of (\ref{Aeq:constraints2}) may hold simultaneously iff
\begin{equation}
a\leq \beta,\quad \beta\neq 0.
\end{equation}

In case of $t>1$ we have $c=b=1-\beta$. The inequalities of (\ref{Aeq:constraints2}) may hold iff
\begin{equation}
a\leq 2\beta-1,\quad \frac{1}{2}\leq\beta<1.
\end{equation}
which completes the proof.


\section*{Acknowledgment}
The authors would like to thank Dr. M.R. Pakravan for suggesting the real curve of the intuitive example.

\ifCLASSOPTIONcaptionsoff
  \newpage
\fi



\bibliographystyle{IEEEtran}
\bibliography{IEEEabrv,VLCrefs}
%

%

%
%
%




\end{document}